\documentclass[12pt, reqno]{amsart}
\usepackage{amsmath, amsthm, amscd, amsfonts, amssymb, graphicx, color, mathrsfs}
\usepackage[bookmarksnumbered, colorlinks, plainpages]{hyperref}
\usepackage{cite}
\usepackage[all]{xy}
\usepackage{slashed}
\usepackage{tikz-cd}
\usepackage{mathabx}
\usepackage{tipa}
\usepackage{soul}
\usepackage{cancel}
\usepackage{ulem}
\usepackage{algorithm}
\usepackage{algpseudocode}
\usepackage{enumitem}
\usepackage{blindtext}

\DeclareMathOperator{\Ad}{Ad}

\DeclareMathOperator{\SU}{SU}
\DeclareMathOperator{\diag}{diag}
\allowdisplaybreaks

\newlength\tindent
\newtheorem{theorem}{Theorem}[section]
\newtheorem{lemma}[theorem]{Lemma}

\newtheorem{proposition}[theorem]{Proposition}
\newtheorem{corollary}[theorem]{Corollary}

\theoremstyle{definition}
\newtheorem{definition}[theorem]{Definition}
\newtheorem{example}[theorem]{Example}

\theoremstyle{remark}
\newtheorem{remark}[theorem]{Remark}
\numberwithin{equation}{section}

\begin{document}
\setcounter{page}{1}

\title[Optimal-Speed Unitary Evolution of Quantum States]{Characterizing Optimal-speed unitary time evolution of pure and quasi-pure quantum states}

\author[John A. Mora Rodríguez]{John A. Mora Rodríguez}
\address{John A. Mora Rodríguez \endgraf
IMECC-Unicamp, Departamento de Matemática Aplicada, Universidade Estadual de Campinas (Unicamp), Rua Sérgio Buarque de Holanda 651, Cidade Universitária Zeferino Vaz. 13083-859, Campinas - SP, Brazil
\endgraf
  {\it E-mail address} {\rm j196970@dac.unicamp.br}
 }

\author[B. Grajales]{Brian Grajales}
\address{Brian Grajales \endgraf
IMECC-Unicamp, Departamento de Matem\'{a}tica,  Universidade Estadual de Campinas (Unicamp). Rua S\'{e}rgio Buarque de Holanda 651, Cidade Universit\'{a}ria Zeferino Vaz. 13083-859, Campinas - SP, Brazil.
\endgraf
{\it E-mail address} {\rm grajales@ime.unicamp.br}
  }

\author[M. {Terra Cunha}]{Marcelo {Terra Cunha}}
\address{Marcelo {Terra Cunha} \endgraf
IMECC-Unicamp, Departamento de Matemática Aplicada,  Universidade Estadual de Campinas (Unicamp). Rua S\'{e}rgio Buarque de Holanda 651, Cidade Universit\'{a}ria Zeferino Vaz. 13083-859, Campinas - SP, Brazil.
\endgraf
{\it E-mail address} {\rm tcunha@unicamp.br}
  }

\author[L. Grama]{Lino Grama}
\address{Lino Grama \endgraf
IMECC-Unicamp, Departamento de Matemática,  Universidade Estadual de Campinas (Unicamp). Rua S\'{e}rgio Buarque de Holanda 651, Cidade Universit\'{a}ria Zeferino Vaz. 13083-859, Campinas - SP, Brazil.
\endgraf
{\it E-mail address} {\rm lino@ime.unicamp.br}
  }



\begin{abstract}

We present a characterization of the Hamiltonians that generate optimal-speed unitary time evolution and the associated dynamical trajectory, where the initial states are either pure states or quasi-pure quantum states. We construct the manifold of pure states as an orbit under the conjugation action of the Lie group $\SU(n)$ on the manifold of one-dimensional orthogonal projectors, obtaining an isometry with the flag manifold $\SU(n)/\textnormal{S}(\textnormal{U}(1)\times \textnormal{U}(n-1 ))$. From this construction, we show that Hamiltonians generating optimal-speed time evolution are fully characterized by equigeodesic vectors of $\SU(n)/\textnormal{S}(\textnormal{U}(1)\times \textnormal{U}(n-1))$. We later extend that result to quasi-pure quantum states.
\end{abstract}

\maketitle

\tableofcontents
\allowdisplaybreaks

\section{Introduction}

Let us consider $\mathcal{H}_n$ an $n$-dimensional complex Hilbert space and $\{|0\rangle,|1\rangle,\ldots,|n-1\rangle\}$ a basis for $\mathcal{H}_n$. Each unitary vector $|\phi\rangle\in\mathcal{H}_n$ is called a vector quantum state. The 
 more general definition of a quantum state is not limited to unitary vectors of the Hilbert space, each positive semi-definite operator of trace 1 on $\mathcal{H}_n$ is also called a quantum state. The set of one-dimensional projectors and the set of projective rays of vector quantum states are biunivocally related. We refer to the elements of these two sets as pure states, depending on the context. We will use the term vector state to refer to vectors and quantum states for operators. We denote by $\langle \phi|$ the linear functional associated with state $|\phi\rangle$ by the inner product. These definitions are well known in the usual formulation of quantum mechanics \cite{Cohen, Nielsen, Peres, BZ, Terra}. We know that if we take a basis for the Hilbert space, each state is represented by a unit vector in $\mathbb{C}^n$. Let us take a basis for $\mathcal{H}_n$ such that $|0\rangle=(1,\vec{0})$. \\

It is well known that for an isolated system, the time evolution of an initial vector state is determined by the Schrödinger equation \cite{Cohen, BZ}.

\begin{equation}
     \displaystyle\frac{d}{dt}|\phi(t)\rangle=\frac{H}{i\hslash
}|\phi(t)\rangle,
\end{equation}

where $H$ is a time-independent Hamiltonian operator on $\mathcal{H}_n$ and $\hslash$ is the Planck's constant. Therefore, the equation $|\phi(t)\rangle=\exp\left(\frac{-iHt}{\hslash}\right)|\phi(0)\rangle$ determines the dynamical trajectory of a state $|\phi(0)\rangle$. Since the evolution must conserve the normalization of the states, we consider $H$ to be Hermitian. Additionally, we disregard the phases in the unitary transformation generated by $H$; that is, we will consider $H$ a traceless operator. This concept of unitary time evolution is naturally extended to quantum states, for which the dynamic trajectory of a quantum state $\rho$ is determined by the equation

\begin{equation}
    \rho(t)=\exp\left(\frac{-iHt}{h}\right)\rho \exp\left(\frac{iHt}{h}\right).
\end{equation}

If we consider two vector states such that one evolves into the other in a given time $T$, that is, for $|\phi\rangle, |\psi\rangle$ we have $|\psi\rangle=\exp\left(\frac{-iHT}{\hslash
}\right)|\phi\rangle$ for some Hamiltonian $H$, we could then ask, what is the smallest value of $T$?, that is, what is the shortest time to transform $|\phi\rangle$ into the vector state $|\psi\rangle$? 
This is a problem that is known as \textit{quantum speed limits} or simply QSL \cite{MT,ML}.\\

In this work, we focus on the complementary problem: which Hamiltonians determine the unitary time evolution that transforms $|\phi\rangle$ into $|\psi\rangle$ in the shortest possible time? The time evolution generated by such types of Hamiltonians is known as optimal-speed unitary time evolution \cite{Carlini,Vaidman,BH, Mo, Brody, Bender, Cafaro, RCB}. Several works in this area share a common focus: they highlight the importance of the geometric structure of the vector state set in analyzing quantum time evolution, in particular the study of geodesics on the Bloch sphere \cite{Mo,Cafaro}. It is also important to note that we are considering the most general form of the problem, assuming Hamiltonians without restrictions. The case where the Hamiltonian that generates unitary time evolution is subject to certain constraints to ensure practical implementation is a well-known variational problem \cite{Carlini}.\\

A criterion with which we can find Hamiltonians generating optimal-speed unitary time evolution was provided by Mostafazadeh \cite{Mo}, establishing that the evolution time $T$ between two vector states $|\phi\rangle, |\psi\rangle$ is determined by: 

\begin{equation}\label{eq:OSTE}
        T=\frac{\hslash
s}{\Delta E_{\phi}(H)},
    \end{equation}

where $s$ is the geodesic distance (in the Fubini-Study metric \cite{AA}) between $|\phi\rangle$ and $|\psi\rangle$, and  $\Delta E_{\phi}(H):=\left[\langle\phi|H^2|\phi\rangle-\langle\phi|H|\phi\rangle^2\right]^{1/2}$ is the energy uncertainty of $H$ in the initial vector state $|\phi\rangle$. We denote by $\Delta E_{\max}(H)$ the maximum energy uncertainty of $H$. In that paper, Mostafazadeh also determines the speed of evolution and indicates that, to achieve maximum speed, the choice of $H$ must be such that $\Delta E_{\phi}(H)=\Delta E_{\max}(H)$.\\

In this paper, we provide a characterization of the maximum energy-uncertainty Hamiltonians that connects two given pure (or quasi-pure) quantum states.
For this purpose, we use the fact that the set of quantum states has a Riemannian structure in which a Lie group acts transitively. 
These manifolds, known as homogenous spaces, have been extensively studied \cite{Berger,Kostant,KV}. 
One significant advantage of working in this context is that, when considering Riemannian metrics that preserve the symmetries given by a Lie group, many geometric problems can be reduced to solving algebraic equations on the Lie algebra associated to the Lie group. 
For example, one can determine infinitesimal generators for homogeneous equigeodesics (see Proposition \ref{criterion}), which are vectors in the Lie algebra whose one-parameter subgroup orbits are geodesics for each Riemannian metric preserving the symmetries of the Lie group.
The notion of equigeodesic is important for classical mechanics, since homogeneous geodesics describe relative equilibrium of the dynamical system, represented by an invariant metric, and we consider the case when this property is stable under changes of the metric.
\\

Building on this foundation, we introduce a new interpretation of equigeodesics, offering a novel perspective on their role in quantum theory. We will use the criterion expressed in equation \eqref{eq:OSTE} and the study of the relationship between homogeneous spaces and sets of quantum states to explore the connection between equigeodesics and dynamic trajectories of optimal-speed unitary time evolution. This approach not only bridges abstract geometric concepts with quantum dynamics but also opens up new ways to understand the basic nature of quantum state changes.

\section{Preliminaries: homogeneous spaces}
\subsection{Invariant metrics} Consider a homogeneous space of the form $G/K,$ where $G$ is a compact finite-dimensional Lie group and $K$ is a closed subgroup of $G.$ A Riemannian metric $g$ defined on $G/K$ is called {\it $G$-invariant} (or {\it $G$-homogeneous}) if the map 
\begin{center}
    $\begin{array}{rccc}
\phi_a: & (G/K,g) & \longrightarrow & (G/K,g)\\
& xK & \longmapsto 
 & axK
\end{array}$
\end{center}
is an isometry for each $a\in G.$ Let $\mathfrak{g}$ and $\mathfrak{k}$ be the Lie algebras of $G$ and $K,$ respectively, and fix an $\Ad^G$-invariant inner product $(\cdot,\cdot)$ on $\mathfrak{g},$ that is, $$(\Ad(a)X,\Ad(a)Y)=(X,Y),\ a\in G,\ X,Y\in\mathfrak{g},$$ where ``$\Ad$" denotes the adjoint representation of $G.$ If $\mathfrak{m}$ is the $(\cdot,\cdot)$-orthogonal complement of $\mathfrak{k}$ in $\mathfrak{g},$ then the decomposition $\mathfrak{g}=\mathfrak{k}\oplus\mathfrak{m}$ is reductive, meaning that $\Ad(k)\mathfrak{m}=\mathfrak{m},$ for all $k\in K.$ This allows us to define a smooth representation \begin{equation*}
    \begin{array}{rccl}
    \Ad^K\big{|}_{\mathfrak{m}}: & K & \longrightarrow & \textnormal{GL}(\mathfrak{m})\\
    & k & \longmapsto & \Ad(k)\big{|}_{\mathfrak{m}}
    \end{array}.
\end{equation*} By compactness of $K$, we have that this representation is completely reducible, which means that there exist subspaces $\mathfrak{m}_1,...,\mathfrak{m}_s$ such that $$\mathfrak{m}=\mathfrak{m}_1\oplus\cdots\oplus\mathfrak{m}_s,\ \textnormal{and}\ \Ad(k)\mathfrak{m}_j=\mathfrak{m}_j,\ \forall k\in K,\ \forall j\in\{1,...,s\}.$$

In this context, the tangent space $T_{eK}(G/K)$ at the origin $eK$ can be identified with $\mathfrak{m}$ via the isomorphism $$\displaystyle\mathfrak{m}\ni X\longmapsto X^*(eK):=\left.\frac{d}{dt}\exp(tX)K\right|_{t=0}\in T_{eK}(G/K),$$ where ``$\exp$" is the exponential map of $G.$ It is well known that the set of $G$-invariant metrics on $G/K$ is in bijection with the set of $\Ad^K$-invariant inner products on $\mathfrak{m},$ which in turn is in bijection with the set of positive, $(\cdot,\cdot)$-selfadjoint and $\Ad^K$-equivariant operators $\Lambda:\mathfrak{m}\longrightarrow\mathfrak{m}.$ An operator with these properties is usually called a {\it metric operator}. 

\subsection{Homogeneous equigeodesics} A smooth path $\gamma$ on $G/K$ is a {\it homogeneous curve} through $aK$ if there exist $X\in\mathfrak{g}$ such that
\begin{equation*}
    \gamma(t)=\exp(tX)aK,\ \forall t\in \mathbb{R}.
\end{equation*}
A {\it homogeneous equigeodesic} through $aK$ is a homogeneous curve through $aK$ which is a geodesic with respect to any $G$-invariant metric on $G/K.$ In this case, we say that $X\in\mathfrak{g}$ is an {\it equigeodesic vector} generating an equigeodesic through $aK.$

\begin{proposition}{\label{prop:equigeodesic}}
    Let $X\in \mathfrak{g}$ be an equigeodesic vector generating an equigeodesic curve through the origin, then $\Ad(a)X$ is an equigeodesic vector generating an equigeodesic curve through $aK.$ 
\end{proposition}
\begin{proof} Let $\gamma_1(t):=\exp(tX)K,$ $\gamma_2(t):=\exp(t\Ad(a)X)aK,$ and assume that $\gamma_1$ is a geodesic with respect to any $G$-invariant metric. Since $\phi_a$ is an isometry for any $G$-invariant metric, then $\phi_a\circ \gamma_1$ is also a geodesic with respect to any $G$-invariant metric. However
\begin{align*}
    (\phi_a\circ \gamma_1)(t)&=\phi_a(\exp(tX)K)\\
    &=a\exp(tX)K\\
    &=a\exp(tX)a^{-1}aK\\
    &=\exp(t\Ad(a)X)aK\\
    &=\gamma_2(t).
\end{align*}
From this, it follows that $\gamma_2$ is an equigeodesic through $aK$ whenever $\gamma_1$ is an equigeodesic through $eK.$ The proof is complete.

\end{proof}
The proposition below provides an effective criterion for determining whether a vector $X\in\mathfrak{g}$ is an equigeodesic vector. This criterion is a slightly modified version of the one established by Cohen, Grama, and Negreiros \cite[Proposition 3.5]{CGN}. For a proof, we refer to \cite[Corollary 3.3]{GG}. Before presenting the result, let us set some notations: for a vector $X\in\mathfrak{g}$ and a linear subspace $V\subseteq \mathfrak{g},$ we will denote $X_V$ the $(\cdot,\cdot)$-orthogonal projection of $X$ onto $V.$

\begin{proposition}\label{criterion}
     A vector $X\in\mathfrak{g}$ is an equigeodesic vector generating an equigeodesic through the origin if and only if \begin{equation}\label{CGN_formula}
            \left[X,\Lambda X_{\mathfrak{m}}\right]_{\mathfrak{m}}=0,
        \end{equation}
        for any metric operator $\Lambda.$ Here $[\cdot,\cdot]$ is the Lie bracket of $\mathfrak{g}.$
\end{proposition}

\subsection{Flag manifolds} In this section, we consider the flag manifold \begin{equation}\label{quotient}\mathbb{F}(n;n_1,...,n_t):=\textnormal{SU}(n)/\textnormal{S}(\textnormal{U}(n_1)\times\cdots\times\textnormal{U}(n_t)),\end{equation} where $n_1+\cdots+n_t=n.$ We will use the notation $[U]=U\textnormal{S}(\textnormal{U}(n_1)\times\cdots\times\textnormal{U}(n_t)),$ for the class of $U\in\textnormal{SU}(n)$ in the quotient \eqref{quotient}. The Lie algebra $\mathfrak{su}(n)$ of $\textnormal{SU}(n)$ consists of skew-Hermitian matrices of order $n.$ We define $(\cdot,\cdot):=-B,$ where $B$ is the Killing form of $\mathfrak{su}(n).$ A reductive $(\cdot,\cdot)$-orthogonal decomposition is given by $$\mathfrak{su}(n)=\mathfrak{s}(\mathfrak{u}(n_1)\oplus\cdots\oplus\mathfrak{u}(n_t))\oplus\mathfrak{m},$$
where $\mathfrak{s}(\mathfrak{u}(n_1)\oplus\cdots\oplus\mathfrak{u}(n_t))$ is the set of block-diagonal matrices 

$$\left(\begin{array}{cccc}X_{11}&0_{n_1\times n_2} & \cdots & 0_{n_1\times n_t}\\ 0_{n_2\times n_1}& X_{22} & \cdots & 0_{n_2\times n_t}\\\vdots & \vdots & \ddots & \vdots\\ 0_{n_t\times n_1}&0_{n_t\times n_2}&\cdots&X_{tt}\end{array}\right),$$
with $X_{jj}^*=-X_{jj}$ and $\sum\limits_{j=1}^{t}\textnormal{Tr}(X_{jj})=0;$ and $\mathfrak{m}$ is the set of matrices of the form
\begin{equation}\label{tangent:vector}\left(\begin{array}{ccccc}0_{n_1\times n_1}&-X_{21}^*&-X_{31}^*&\cdots&-X_{t1}^*\\
X_{21}&0_{n_2\times n_2}&-X_{32}^*&\cdots&-X_{t2}^*\\
X_{31}&X_{32}&0_{n_3\times n_3}&\cdots&-X_{t3}^*\\
\vdots&\vdots&\vdots&\ddots&\vdots\\
X_{t1}&X_{t2}&X_{t3}&\cdots&0_{n_t\times n_t}\end{array}\right).
\end{equation}
The decomposition of $\mathfrak{m}$ into irreducible  $\Ad^{\textnormal{S}(\textnormal{U}(n_1)\times\cdots\times\textnormal{U}(n_t))}$-invariant subspaces is
$$\mathfrak{m}=\bigoplus\limits_{1\leq j<i\leq t}\mathfrak{m}_{ij},$$ where $\mathfrak{m}_{ij}$ is the set of matrices of the form \eqref{tangent:vector} such that $X_{rs}=0_{n_r\times n_s},$ for $(r,s)\neq (i,j).$ It is well known that these subspaces $\mathfrak{m}_{ij}$ are pairwise inequivalent, consequentely, any metric operator $\Lambda$ is determined by a family $\{\mu_{ij}:1\leq j< i\leq t\}$ of positive numbers, such that $\Lambda$ applied to a matrix of the form \eqref{tangent:vector} multiplies each block $X_{ij}$ by $\mu_{ij}.$

\begin{proposition}[\cite{CGN}]\label{CGN:theorem} 
A matrix of the form \eqref{tangent:vector} is an equigeodesic vector generating an equigeodesic through $[\textnormal{Id}]$ if and only if $X_{ij}X_{jk}=0_{n_i\times n_k}$ for all distinct $i,j,k.$    
\end{proposition}

A valuable particular case appears when we consider $n_1=1$ and $n_2=n-1,$ so the quotient \eqref{quotient} becomes

$$\mathbb{F}(n;1,n-1)=\textnormal{SU}(n)/\textnormal{S}(\textnormal{U}(1)\times\textnormal{U}(n-1)).$$ 

In this case, a vector $X\in\mathfrak{su}(n)$ can be written as

\begin{equation}\label{general:vector}X=\left(\begin{array}{cc}i\alpha&-x^*\\
x&A_{(n-1)\times (n-1)}\end{array}\right),
\end{equation}

where $\alpha$ is a real number, $A^{*}=-A$ and $\textnormal{Tr}(A)=-i\alpha.$ The Lie algebra $\mathfrak{s}(\mathfrak{u}(1)\oplus\mathfrak{u}(n-1))$ consists of all matrices of the form \eqref{general:vector} with $x=0,$ and its $(\cdot,\cdot)$-orthogonal complement $\mathfrak{m}$ is the space of matrices of the form \eqref{general:vector} such that $A=\alpha=0.$ Additionally, $\mathfrak{m}$ is $\Ad^{\textnormal{S}(\textnormal{U}(1)\times\textnormal{U}(n-1))}$-irreducible, thus for  any metric operator $\Lambda,$ there exists $\mu>0$ such that $\Lambda=\mu\textnormal{Id}_{\mathfrak{m}}.$  
\begin{proposition}\label{general:equigeodesics} A matrix $X\in\mathfrak{su}(n)$ written as in \eqref{general:vector} is an equigeodesic vector generating an equigeodesic through $[\textnormal{Id}]$ if and only if $Ax=i\alpha x.$ 
\end{proposition}
\begin{proof} Consider $X\in\mathfrak{su}(n)$ written as in \eqref{general:vector}. If $A=0$ then the result follows from Proposition \ref{CGN:theorem}. Otherwise, by virtue of Proposition \ref{criterion}, $X$ is an equigeodesic vector generating an equigeodesic through $[\textnormal{Id}]$ if and only if $[X,\Lambda X_{\mathfrak{m}}]_{\mathfrak{m}}=0$ for any metric operator $\Lambda.$ Since all metric operators are positive multiples of the identity map on $\mathfrak{m},$ we have that  $[X,\Lambda X_{\mathfrak{m}}]_{\mathfrak{m}}=0$ if and only if $\mu[X,X_{\mathfrak{m}}]_{\mathfrak{m}}=0,$ for all $\mu>0,$ that is, $[X,X_{\mathfrak{m}}]_{\mathfrak{m}}=0.$ Considering that $$X_{\mathfrak{m}}=\left(\begin{array}{cc}0&-x^*\\
x&0_{(n-1)\times (n-1)}\end{array}\right),$$ we compute $[X,X_{\mathfrak{m}}]_{\mathfrak{m}}:$
\begin{align*}
    [X,X_{\mathfrak{m}}]_{\mathfrak{m}}=&\left(\begin{array}{cc}i\alpha&-x^*\\
x&A_{(n-1)\times (n-1)}\end{array}\right)\left(\begin{array}{cc}0&-x^*\\
x&0_{(n-1)\times (n-1)}\end{array}\right)\\
\\
&-\left(\begin{array}{cc}0&-x^*\\
x&0_{(n-1)\times (n-1)}\end{array}\right)\left(\begin{array}{cc}i\alpha&-x^*\\
x&A_{(n-1)\times (n-1)}\end{array}\right)\\
\\
=& \left(\begin{array}{cc}0& x^*A-i\alpha x^*\\ Ax-i\alpha x&0\end{array}\right)\\
\\
=& \left(\begin{array}{cc}0& -(Ax-i\alpha x)^*\\ Ax-i\alpha x&0\end{array}\right).
\end{align*}
Thus, $[X,X_{\mathfrak{m}}]_{\mathfrak{m}}=0$ if and only if $Ax-i\alpha x=0.$ This completes the proof. 

\end{proof}

\begin{corollary}\label{corollary:quigeodesics} The equigeodesic curve through $[\textnormal{Id}]$ generated by an equigeodesic vector $X\in\mathfrak{su}(n)$ is the same as the equigeodesic curve thorugh $[\textnormal{Id}]$ generated by its projection $X_{\mathfrak{m}}$ onto $\mathfrak{m}.$
\end{corollary}
\begin{proof} Let $$X=\left(\begin{array}{cc}i\alpha&-x^*\\
x&A_{(n-1)\times (n-1)}\end{array}\right),$$ be an equigeodesic vector generating an equigeodesic curve  through $[\textnormal{Id}].$ Then, by Proposition \ref{general:equigeodesics}, we have that $Ax=i\alpha x$. This implies that $$X_{\mathfrak{s}(\mathfrak{u}(1)\oplus\mathfrak{u}(n-1))}=\left(\begin{array}{cc}i\alpha&0_{1\times(n-1)}\\
0_{(n-1)\times 1}&A\end{array}\right)$$ commutes with $$X_{\mathfrak{m}}=\left(\begin{array}{cc}0&-x^*\\
x&0_{(n-1)\times (n-1)}\end{array}\right).$$
Therefore
\begin{align*}
    [\exp(tX)]=&[\exp(tX_{\mathfrak{m}}+tX_{\mathfrak{s}(\mathfrak{u}(1)\oplus\mathfrak{u}(n-1))})]\\
    =&[\exp(tX_{\mathfrak{m}})\exp(tX_{\mathfrak{s}(\mathfrak{u}(1)\oplus\mathfrak{u}(n-1))})]\\
    =&[\exp(tX_{\mathfrak{m}})],
\end{align*}
 where the last equality holds because $\exp(tX_{\mathfrak{s}(\mathfrak{u}(1)\oplus\mathfrak{u}(n-1))})\in\textnormal{S}(\textnormal{U}(1)\times\textnormal{U}(n-1)),$ for all $t\in\mathbb{R}.$ The proof is complete.  
\end{proof}

\begin{remark} The adjoint representation of $\textnormal{SU}(n)$ is given by $\Ad(U)X=UXU^*,$ so an equigeodesic vector generating an equigeodesic through $[U]$ is nothing but the conjugation by $U$ of a matrix of the form \eqref{general:vector} with $Ax=i\alpha x$ (see Proposition \ref{prop:equigeodesic}).
\end{remark}
\section{Characterization of optimal-speed unitary time evolution}

In this section, we introduce the construction of the manifold of pure states as the orbit of the 
$\textnormal{SU}(n)$ Lie group acting on it. This action establishes an isometric relationship between this manifold and the flag manifold $\mathbb{F}(n;1,n-1)$. This connection allows for a characterization of the Hamiltonians that generate optimal-speed unitary time evolution.

\subsection{Quantum states and flag manifolds}

We denote by $D_n$ the set of quantum states over the Hilbert space $\mathcal{H}_n$. Let $D_{nk}$ denote the subset of $D_n$ consisting of all quantum states of rank $k\leq n$. We have that each $D_{nk}$
has the structure of a Riemannian manifold, and $D_n$
is the union of these manifolds $D_{nk}$, although it is merely a convex set \cite{Dit}.\\

Let 

\begin{align*}
    \theta: \SU(n) \times D_{nk} &\rightarrow D_{nk}\\
    (U,\rho) & \mapsto U\rho U^* 
\end{align*}

denote the conjugation action of $\SU(n)$ on $D_{nk}$. Note that if $\rho\in D_{nk}$, we can write $\rho$ as follows: $\rho=\sum_{i=1}^k p_i|\phi_i\rangle\langle\phi_i|$ where $\{|\phi_i\rangle\}_{i=1}^k$ is an orthonormal subset of $\mathcal{H}_n$. Then, $\rho=UDU^*$ with $U\in \SU(n)$ and $D=\diag(p_1,\ldots,p_k,0,\ldots,0)$, therefore $\rho\in \theta_D$ the orbit of the action $\theta$ over $D$.

\begin{proposition}
    Let $\rho=\sum_{i=1}^k p_i|\phi_i\rangle\langle\phi_i|$ be a quantum state over $\mathcal{H}_n$. Then $\rho$ belongs to a Riemannian manifold isometric to the flag manifold $$\mathbb{F}(n;n_1,...,n_r,n-k),$$ where $\sum_{j=1}^r n_j=k$. 
\end{proposition}

\begin{proof}
    We have that 
    
    $$\rho\in \theta_{D}=\left\{\theta(U,D):U\in SU(n)\right\},$$
    
    where $D=\diag(p_1,\ldots,p_k,0,\ldots,0)$. Because $D_{nk}$ is a Riemannian manifold, then the orbit $\theta_{D}$ is a Riemannian manifold diffeomorphic to $\SU(n)/I_D$, where $I_D:=\{U\in \SU(n): \theta(U,D)=D\}$ is the isotropy group of $D$. Therefore, there exists a diffeomorphism $f$ from $\SU(n)/I_{D}$ to $\theta_D.$ If we take a Riemannian metric $g$ on $\theta_D$, we can equip $\SU(n)/I_D$ with a Riemannian metric such that $f$ becomes an isometry.\\

    Finally, it is not difficult to see that $I_D =\textnormal{S}(\textnormal{U}(n_1)\times\cdots\times\textnormal{U}(n_r)\times\textnormal{U}(n-k))$ where each $n_j$ is the multiplicity of each eigenvalue of $\rho$.

\end{proof}
    
\begin{example}\label{example:pure and quasi-pure estates}

As we mentioned before, we will consider a basis for the Hilbert space $\mathcal{H}_n$ such that $|0\rangle$ is represented by the complex vector $(1,\vec{0})^T$.\\

    \begin{enumerate}
        \item Let $D=|0\rangle\langle 0|=\diag(1,0,0,\ldots,0)$. The orbit of the action $\theta$ in $D$ is the set of pure states, that is, $D_{n1}$ is isometric to $\mathbb{F}(n;1,n-1)$.\\

        \item If $D=\diag(p_1,p_2,p_2,\ldots,p_2)$, then the orbit $\theta_D$ is also isometric to the flag manifold $\mathbb{F}(n;1,n-1)$.\\
        
    \end{enumerate}
    
Note that different orbits can be isometric to the same flag manifold. We discuss this relationship at the end of this section.

\end{example}

\subsection{Main results} 

Now that we have prepared the way, we are two lemmas away from presenting a comprehensive characterization of optimal-speed time evolution for pure states in any dimension. Furthermore, we will see that the evolution of a one-dimensional subspace takes the orthogonal subspace along with it, therefore we extend these findings to a specific class of quantum states known as quasi-pure states (Definition \ref{def:quasi-pure}).

\begin{lemma}{\label{lemma:equig}}

    Let $f:\mathbb{F}(n;1,n-1)\rightarrow D_{1n}$ be the diffeomorphism defined by $f([U])=U|0\rangle\langle 0| U^*$. Then the image of equigeodesic curves through the origin in $\mathbb{F}(n;1,n-1)$ by $f$ determines dynamic trajectories of optimal-speed unitary time evolution.\\

    In other words, the Hamiltonian $H=iX$, with $X\in\mathfrak{m}$ an equigeodesic vector for $$\mathbb{F}(n;1,n-1),$$ generates optimal-speed unitary time evolution from the initial vector state $|0\rangle$.
\end{lemma}

\begin{proof}

    Let us see that the Hamiltonian

$$H=iX=\begin{pmatrix}
 0 & x^*\\
 x & 0_{n-1\times n-1}
\end{pmatrix}$$

generates optimal-speed time evolution starting at the vector state $|0\rangle$. Observe that 

\begin{equation}
    \Delta E_0(H)^2=x^*x=\Arrowvert x\Arrowvert_2^2,
\end{equation}

and that for all $|\phi\rangle=(a, \vec{b})^T$, it follows from the Cauchy-Schwarz inequality that

\begin{equation}
    \Delta E_{\phi}(H)^2=|a|^2\Arrowvert x \Arrowvert_2^2+|x^*b|^2-4Re(\Bar{a}x^*b)^2\leq \Arrowvert x \Arrowvert_2^2-4Re(\Bar{a}x^*b)^2.
\end{equation}

Therefore, $\Delta E_0(H)$ is maximum.

\end{proof}

\begin{lemma}{\label{lemma:propH}}
    Let
$$H=\begin{pmatrix}
 \alpha & x^*\\
 x & A_{n-1\times n-1}
\end{pmatrix}$$

be a Hamiltonian with $A\neq 0$ and $Ax\neq \alpha x$. Then, $\Delta E_{max}(H)>\Arrowvert x\Arrowvert_2$.
\end{lemma}

\begin{proof}
    Suppose that $x$ is not an eigenvector of $A$. Let $|\phi\rangle=(0, x)^T/\Arrowvert x\Arrowvert_2$, then

    \begin{equation}
        \Delta E_{\phi}(H)^2=\Arrowvert x \Arrowvert_2^2+\left(\Delta E_{\Tilde{\phi}}(A)\right)^2>\Arrowvert x\Arrowvert_2^2,
    \end{equation}

    with $|\Tilde{\phi}\rangle=x/\Arrowvert x\Arrowvert_2$. Now, let us assume that $Ax=\zeta x$ with $\zeta\neq \alpha$. As 

    \begin{equation}
        \det\begin{pmatrix}
    \alpha-\lambda & x^*\\
    x & A-\lambda I
\end{pmatrix}=\det(A-\lambda I)\det(\alpha-\lambda-x^*(A-\lambda I)^{-1}x)
    \end{equation}

    we have that the the roots of $p(\lambda)=\alpha-\lambda-\Arrowvert x\Arrowvert_2^2/(\zeta-\lambda)$ are eigenvalues of $H$, this is 

    \begin{equation}
        \lambda_{+,-}=\frac{1}{2}\left(\pm\sqrt{(\alpha-\zeta)^2+4\Arrowvert x\Arrowvert_2^2}+\alpha+\zeta\right),
    \end{equation}

    are eigenvalues of $H$. By the Popoviciu inequality $\Delta E_{max} (H)\leq (\max(spec(H))-\min(spec(H)))/2$ \cite{BD}, and considering the energy uncertainty of the Hamiltonian in the uniform mixture of the eigenvectors of $H$ associated with the extreme eigenvalues of the spectrum, we have: $\Delta E_{max} (H)= (\max(spec(H))-\min(spec(H)))/2$; therefore 

    \begin{equation}
        \Delta E_{max} (H)\geq \frac{\lambda_+-\lambda_-}{2}>\Arrowvert x\Arrowvert_2.
    \end{equation}

\end{proof}

\begin{theorem}{\label{teo:characterization}}
    Let $f:\mathbb{F}(n;1,n-1)\rightarrow D_{1n}$ be the diffeomorphism defined by $f([U])=U|0\rangle\langle 0| U^*$. Then every dynamic trajectory generated by an optimal-speed unitary time evolution, with initial vector state $|0\rangle$, is the image of an equigeodesic curve  through the origin in $\mathbb{F}(n;1,n-1)$ by $f$.\\

    In other words, every Hamiltonian that generates optimal-speed unitary time evolution with initial vector state $|0\rangle$ is given by an equigeodesic vector for $\mathbb{F}(n;1,n-1)$ multiplied by $i$.
\end{theorem}

\begin{proof}
    Let us observe that for every Hamiltonian 
    
    $$H=\begin{pmatrix}
 \alpha & x^*\\
 x & A_{n-1\times n-1}
\end{pmatrix},$$

    we have to $\Delta E_0(H)=\Arrowvert x\Arrowvert_2$. Therefore, by Lemma \ref{lemma:propH}, $H$ generates optimal-speed unitary time evolution whenever $A=0$ or $Ax=\alpha x$. In the first case, 

    \begin{equation}
        H=\begin{pmatrix}
 0 & x^*\\
 x & 0_{n-1\times n-1}
\end{pmatrix}.
    \end{equation}

    In the second case, by Corollary \ref{corollary:quigeodesics}, the equigeodesic curve generated by the equigeodesic vector $-iH$ is the same as the equigeodesic curve generated by

     \begin{equation}
        X=\begin{pmatrix}
 0 & -ix^*\\
 -ix & 0_{n-1\times n-1}
\end{pmatrix}.
    \end{equation}

    Therefore, the results follows from Lemma \ref{lemma:equig}.

\end{proof}

So far, we have established a characterization for optimal-speed time evolution when the initial vector state is $|0\rangle$. In the following result, we generalize the previous findings and present an independent characterization of the initial vector state.

\begin{corollary}\label{cor:characterization}
    Every Hamiltonian that generates optimal-speed unitary time evolution with an arbitrary initial vector state $|\phi\rangle$ is given by an equigeodesic vector for $\mathbb{F}(n;1,n-1)$ multiplied by $i$.
\end{corollary}

\begin{proof}

    Let $|\phi\rangle$ be a vector state. For the quantum state we have that $|\phi\rangle\langle\phi|=U|0\rangle\langle 0|U^*,$ with $U\in \SU(n)$. Then

    \begin{align*}
    \Delta E_\phi(H)^2&=Tr(H^2|\phi\rangle\langle\phi|)-Tr(H|\phi\rangle\langle\phi|)^2\\
    & = Tr(H^2U|0\rangle\langle0|U^*)-Tr(HU|0\rangle\langle 0|U^*)^2\\
    & = Tr((U^*HU)^2|0\rangle\langle0|)-Tr(U^*HU|0\rangle\langle 0|)^2.\\
    &=\Delta E_0(U^*HU)
\end{align*}

Furthermore, given $|\psi\rangle\neq |\phi\rangle$, let $V\in \SU(n)$ such that $V|0\rangle=|\psi\rangle$.   

\begin{align*}
    \Delta E_\psi(H)^2& = Tr(H^2V|0\rangle\langle0|V^*)-Tr(HV|0\rangle\langle 0|V^*)^2\\
    & = Tr((U^*HU)^2W|0\rangle\langle0|W^*)-Tr((U^*HU)W|0\rangle\langle 0|W^*)^2.\\
    &=\Delta E_\varphi(U^*HU),
\end{align*}

with $W=U^*V$ and $|\varphi\rangle=W|0\rangle$. Therefore, by Theorem \ref{teo:characterization} it follows that $U^*HU=iX$ with $X$ equigeodesic vector generating an equigeodesic curve through the origin, that is, $H=iUXU^*$. From Proposition \ref{prop:equigeodesic}, it follows that $H=i\Tilde{X},$ where $\Tilde{X}$ is an equigeodesic vector generating an equigeodesic curve through $U$. 

\end{proof}

Finally, we are going to define a special type of quantum states in $D_{nn}$ for which it becomes nearly evident that we can extend the result in Corollary \ref{cor:characterization}.

\begin{definition}\label{def:quasi-pure}
    Let $\{|\phi_1\rangle,\ldots,|\phi_n\rangle\} $ be an orthonormal basis of a Hilbert space $\mathcal{H}_n$. The quantum state  

    $$\rho=p_1|\phi_1\rangle\langle\phi_1|+p_2\sum_{i=2}^n|\phi_i\rangle\langle\phi_i|,$$

    with $p_1\neq p_2$, is called the quasi-pure state\footnote{A quantum state is usually called quasi-pure when $p_1>p_2$, but this distinction is not relevant for the discussion here.}.
\end{definition}

\begin{proposition}\label{prop:quasi-pure}
    Let $\rho=p_1|\phi_1\rangle\langle\phi_1|+p_2\sum_{i=2}^n|\phi_i\rangle\langle\phi_i|$, $\varrho=p_1|\psi_1\rangle\langle\psi_1|+p_2\sum_{i=2}^n|\psi_i\rangle\langle\psi_i|$ be quasi-pure states and $U$ an unitary matrix such that $U|\phi_1\rangle=|\psi_1\rangle$. Then

    \begin{equation}
        U\rho U^*=\varrho.
    \end{equation}
\end{proposition}

\begin{proof}
    By hypothesis, $U|\phi_1\rangle=|\psi_1\rangle$. Therefore, 

    \begin{align*}
        U\rho U^*&=p_1U|\phi_1\rangle\langle\phi_1|U^*+p_2\sum_{i=2}^nU|\phi_i\rangle\langle\phi_i|U^*\\
        &=p_1|\psi_1\rangle\langle\psi_1|+p_2\sum_{i=2}^nU|\phi_i\rangle\langle\phi_i|U^*.
    \end{align*}

    Since $U$ is an unitary matrix, $\{U|\phi_i\rangle\}_{i=2}^n$ is an orthonormal set, orthogonal to $|\phi_1\rangle$. Then, $span(\{U|\phi_i\rangle\}_{i=2}^n)=span(\{|\psi_i\rangle\}_{i=2}^n)$, this means

    $$\sum_{i=2}^nU|\phi_i\rangle\langle\phi_i|U^*=\sum_{i=2}^n|\psi_i\rangle\langle\psi_i|.$$

    Hence, $U\rho U^*=\varrho.$ 
    
\end{proof}

The previous result implies that the evolution time between the quasi-pure states $\rho=p_1|\phi_1\rangle\langle\phi_1|+p_2\sum_{i=2}^n|\phi_i\rangle\langle\phi_i|$ and $\varrho=p_1|\psi_1\rangle\langle\psi_1|+p_2\sum_{i=2}^n|\psi_i\rangle\langle\psi_i|$  is the same evolution time between the states $|\phi_1\rangle$ and $|\psi_1\rangle$, that is, the optimal-speed unitary time evolution between quasi-pure states reduces to optimal-speed time evolution between pure states. The relationship between pure states and quasi-pure mixed states lends meaning to the connection between the geometric structures associated with these states, as discussed in Example \ref{example:pure and quasi-pure estates}. This connection allows us to further develop our characterization in the following result.

\begin{corollary}
    Every Hamiltonian that generates optimal-speed unitary time evolution with initial state a quasi-pure state, is an equigeodesic vector for $\mathbb{F}(n;1,n-1)$ multiplied by $i$.
\end{corollary}

\begin{proof}
    Follows from Proposition \ref{prop:quasi-pure} and Corollary \ref{cor:characterization}.
    
\end{proof}

\section{Conclusion and outlook}

In Corollary \ref{cor:characterization}, we present a characterization of the optimal-speed unitary time evolution that provides a deterministic criterion for this type of quantum time evolution. Our result also offers an elegant proof of the geodesicity of the dynamic trajectories generated by Hamiltonians that maximize the speed of evolution between two vector states. In this sense, we complement the result obtained by Cafaro and Alsing \cite{Cafaro}, as we observe that the dynamic trajectories generated by optimal-speed unitary time evolution are isometric images of equigeodesics; that is, the geodesicity of these curves is independent of the metric.\\

Additionally, we know that the geodesic dynamic trajectory connecting two vector states is entirely contained in the subspace generated by these two vector states \cite{Brody, Cafaro, AA}. 
In this sense, the problem for pure states reduces to the search for Hamiltonians generating dynamics on the Bloch sphere generated by the two state vectors.
In this work we show that for larger Hilbert spaces there is a large family of Hamiltonians which connect the two given pure states in optimal time.
In fact, we present the algebraic structure of these operators.\\

Using the theory of differential geometry of flag manifolds to study quantum dynamics reaffirms the fascinating connection between abstract mathematics and quantum theory. 
The geometric structure of flag manifolds offers a way to understand how quantum states evolve. This combination not only enhances our understanding of the mathematical principles involved but also highlights the beauty of how these ideas can help in solving problems in quantum dynamics. 
By linking these fields, we uncover new insights into quantum processes, showing how mathematical concepts can effectively address key questions in quantum theory.\\

In general, determining optimal-speed unitary time evolution and deciding on the geodesicity of their dynamic trajectories for non-pure quantum states is not a straightforward task. This is partly because, for varieties $D_{nk}$ with $k> 1$, we must consider that there exist infinitely many monotone Riemannian metrics \cite{MC, Petz}. Our work paves an interesting path in this direction. Looking ahead, a natural next step in generalizing the ideas of this paper is to consider states in the orbit of $D=\diag(p_1,...,p_1,p_2,...,p_2),$ where $p_1$ appears $k$ times. This is because the associated flag manifold $\SU(n)/\textnormal{S}(\textnormal{U}(k)\times \textnormal{U}(n-k))$ has a reductive decomposition $\mathfrak{su}(n)=\mathfrak{s}(\mathfrak{u}(k)\oplus\mathfrak{u}(n-k))\oplus\mathfrak{m},$ where $\mathfrak{m}$ is $\Ad^{\textnormal{S}(\textnormal{U}(k)\times \textnormal{U}(n-k))}$-irreducible. This decomposition allows us to characterize the equigeodesic vectors in a manner similar to that provided in Proposition \ref{general:equigeodesics}.

\section*{Acknowledgements}
John A. Mora Rodríguez is supported by the Coordenação de Aperfeiçoamento de Pessoal de Nível Superior (CAPES) - Finance Code 001. Brian Grajales is supported by grant 2023/04083-0 (São Paulo Research Foundation FAPESP). Marcelo Terra Cunha acknowledges partial support from CNPq, grant number 311314/2023-6. Lino Grama is partially supported by S\~ao Paulo Research Foundation FAPESP grants 2021/04065-6  and 2023/13131-8. This work is part of the Brazilian National Institute of Science and Technology in Quantum Information.

\bibliographystyle{apa}

\end{document}